\documentclass[conference, final, letterpaper]{IEEEtran}

\usepackage[utf8]{inputenc}
\usepackage{graphicx}
\usepackage{stfloats}
\usepackage{amssymb}
\usepackage[cmex10]{amsmath}
\usepackage{color}
\usepackage{theorem}
\usepackage{pifont}
\usepackage{psfrag}
\usepackage{hyperref}
\usepackage{url}
\usepackage{breakurl}

\usepackage{bbm}
\usepackage{mathdots}
\usepackage{relsize}

\usepackage{ifthen}

\newcommand{\F}{\ensuremath{\mathbb{F}}\,}
\newcommand{\Fx}{\ensuremath{\mathbb{F}[x]}\,}
\newcommand{\Fxsub}[1]{\ensuremath{\mathbb{F}_{#1}[x]}}
\newcommand{\Fxy}{\ensuremath{\mathbb{F}[x,y]}\,}

\newcommand{\refeq}[1]{(\ref{#1})}

\newtheorem{definition}{Definition}
\newtheorem{theorem}{Theorem}

\newtheorem{proposition}{Proposition}

\newcommand{\MHASSE}[2]{\ensuremath{Q^{[#1,#2]}(x,y)}}
\newcommand{\MHASSEB}[4]{\ensuremath{Q^{[#1,#2]}(#3,#4)}}

\newcommand{\QPOL}[1]{\ensuremath{Q^{(#1)}(x)}}
\newcommand{\RS}[2]{\ensuremath{\mathcal{RS}(#1,#2)}}
\newcommand{\VIRS}[1]{\ensuremath{\mathcal{VIRS}(n,k,s)}}

\newcommand{\SET}[1]{\ensuremath{[#1]}}
\newcommand{\SETzero}[1]{\ensuremath{[#1]_0}}

\newcommand{\mat}[2][\empty]{
  \ifthenelse{\equal{#1}{\empty}}
    {\ensuremath{\mathbf{#2}}}
    {\ensuremath{#2_{#1}}}
}

\newcommand{\vect}[2][\empty]{
  \ifthenelse{\equal{#1}{\empty}}
    {\ensuremath{\mathbf{#2}}}
    {\ensuremath{#2_{#1}}}
}

\begin{document}

\title{A Link between Guruswami--Sudan's List--Decoding and Decoding of Interleaved Reed--Solomon Codes}

\IEEEoverridecommandlockouts

\author{\authorblockN{Alexander Zeh and Christian Senger}\thanks{This work has been supported by DFG,
Germany, under grants BO~867/17 and BO~867/22-1.}
\authorblockA{Institute of Telecommunications and Applied Information Theory\\
Ulm University, Germany\\
\texttt{\{alexander.zeh,christian.senger\}@uni-ulm.de}}
}
\maketitle

\begin{abstract}
The Welch--Berlekamp approach for Reed--Solomon (RS) codes forms a bridge between classical syndrome--based decoding algorithms and interpolation--based list--decoding procedures for list size $\ell=1$. It returns the univariate error--locator polynomial and the evaluation polynomial of the RS code as a $y$--root.\\ 
In this paper, we show the connection between the Welch--Berlekamp approach for a specific Interleaved Reed--Solomon code scheme and the Guruswami--Sudan principle. It turns out that the decoding of Interleaved RS codes can be formulated as a modified Guruswami--Sudan problem with a specific multiplicity assignment.
We show that our new approach results in the same solution space as the Welch--Berlekamp scheme. Furthermore, we prove some important properties.
\end{abstract}


\section{Introduction}
The Guruswami--Sudan (GS)~\cite{GuruswamiSudan_ImproveddecodingofReed-Solomonandalgebraic-geometrycodes_1999} approach for Reed--Solomon (RS) codes consists of an interpolation and a factorization step of a degree--restricted bivariate polynomial. The usage of multiplicities in the first stage improved the error--correcting capability of Sudan's original work~\cite{sudan97decoding}. The set of $y$--roots of the bivariate interpolation polynomial gives the candidates of the evaluation polynomials of the corresponding RS codes.
The GS principle coincides with the Welch-Berlekamp (WB) approach~\cite{BerlekampWelch_Patent} when the list size is $\ell=1$. Then, $\tau_0 = \lfloor (n-k)/2 \rfloor$ errors can be uniquely corrected, where $n$ is the length and $k$ the dimension of the RS code.\\
Interleaved Reed--Solomon (IRS) codes are most effective if correlated errors affect all words of the interleaved scheme simultaneously (see~\cite{Krachkovsky_IRS_1998}). Because of this, IRS codes are mainly considered in applications where error bursts occur. Bleichenbacher \textit{et al.}~\cite{Bleichenbacher_DecodingofIRS_2003, Bleichenbacher_DecodingofIRS_2007} formulated an IRS decoding procedure with the WB method.\\ 
Our contribution covers the reformulation of the Bleichenbacher approach in terms of a modified GS interpolation problem for a heterogeneous IRS scheme as it was investigated in~\cite{Schmidt_DecodingReedSolomonCodesBeyondHalf_2006}. The heterogeneous IRS code is built by virtual extension of an RS code. The rate restriction and the decoding radius of this scheme are comparable with the parameters of Sudan's original algorithm (where the multiplicity for each point equals one). Also, the corresponding syndrome formulation (for Sudan done in~\cite{Ruckenstein_PHD2001, Roth_Ruckenstein_2000}) is equivalent. Hence, it seems to be surprising that this scheme can be formulated as a modified GS interpolation problem, where the multiplicities are assigned in a specific manner.\\
The paper is organized as follows. First, we shortly describe the GS principle for RS codes in Section~\ref{sec_GS} and outline important properties that we will use later. The connection to the WB approach is investigated in Section~\ref{sec_listone}. The virtual extension to an IRS code~\cite{Schmidt_DecodingReedSolomonCodesBeyondHalf_2006} is described in Section~\ref{sec_pdec}. Section~\ref{sec_GS-light} links the GS list--decoding procedure with the WB formulation of the previously described IRS scheme. Furthermore, the equivalence of both approaches is proved and an informal description is given. Finally, Section~\ref{sec_future} concludes the paper. An example is given in the appendix.

\section{Definition and Notation} \label{sec_Notation}
Here and later, \SET{n} denotes the set of integers $\{1,\dots,n\}$ and \SETzero{n} denotes the set of integers $\{0,\dots,n\}$. The entries of an $m \times n$ matrix $\mat{S}= \parallel \mat[i,j]{S} \parallel $ are denoted $\mat[i,j]{S} $, where $i \in \SETzero{m-1}$ and $j \in \SETzero{n-1}$. A univariate polynomial of degree $n$ is noted in the form $A(x) = \sum_{i=0}^{n} A_i x^i$. A vector of length $n$ is denoted by $\vect{r} = (\vect[1]{r},\vect[2]{r},\dots,\vect[n]{r})^T$.\\
Let $\alpha_1, \alpha_2, \dots, \alpha_n$ be nonzero distinct elements (code locators) of the finite field $\F=GF(q)$ of size $q$. $\mathcal{L} = \{\alpha_1,\dots, \alpha_n \}$ is the set containing all code locators.
Denote 
\begin{equation*}
 f(\mathcal{L}) = (f(\alpha_i), \dots, f(\alpha_n))
\end{equation*}
for a given polynomial $f(x)$ over $\F$.\\
An RS code \RS{n}{k} over $\F$ with $n<q$ is given by
\begin{equation} \label{eq_defRS}
 \RS{n}{k} = \{ \mathbf{c} = f(\mathcal{L}) : f(x) \in \Fxsub{k} \}, 
\end{equation}
where \Fxsub{k} stands for the set of all univariate polynomials with degree less than $k$ and indeterminate $x$.\\
RS codes are known to be maximum distance separable (MDS), i.e., their minimum Hamming distance is $d=n-k+1$.

\section{The GS principle and the univariate formulation} \label{sec_GS}
\subsection{Guruswami--Sudan Approach for Reed--Solomon Codes}
Let the $n$ points $\lbrace(\alpha_i,r_i)\rbrace _{i=1}^n \, \text{,where} \, \alpha_i,r_i \in \F $ and  $\mathbf{r}=(r_1,\dots, r_n)$ denotes the received word, be interpolated by a bivariate polynomial $Q(x,y)$. The number of errors, that can be corrected, is denoted by $\tau$. The parameter $s$ is the order of multiplicity of the bivariate interpolation polynomial in the GS algorithm. The list size is denoted by $\ell$. The nonzero interpolation polynomial $Q(x,y)$ has to satisfy the following degree conditions:
\begin{equation} \label{eq_GS_cond1}
DC_1 := \begin{bmatrix}
\deg_{0,1}Q(x,y) & \leq  & \ell,\\
\deg_{1,k-1}Q(x,y) & < & s(n-\tau) \\
\end{bmatrix},
\end{equation} 
where $\deg_{u,v} a(x,y) = ud_x + vd_y $ is the $(u,v)$--weighted degree of a bivariate polynomial $a(x,y) = \sum_{i=0}^{d_x} \sum_{j=0}^{d_y} a_{i,j} x^iy^j$.
The interpolation constraints are:
\begin{equation} \label{eq_interpol2}
IC_1 :=  \begin{bmatrix} \MHASSEB{a}{b}{\alpha_i}{r_i} = 0 \quad \forall i \in \SET{n} \, \text{and} \, \forall a+b < s \end{bmatrix},
\end{equation}
where $\MHASSE{a}{b}$ represents the mixed Hasse derivative (see~\cite{Hasse_1936} for definition) of the polynomial $Q(x,y) \in \Fxy$. Analogously, one can say that the GS polynomial must have a multiplicity of $s$ at each point $(\alpha_i, r_i)$.

\subsection{Univariate Formulation of Guruswami--Sudan}
In~\cite{AugotZeh_OnTheRREquations_2008, ZehA_AnBerlekampMassey_2010} the univariate reformulation of the bivariate GS interpolation problem (\textit{key equations}) was derived. Here, we state some basic properties that will be used later on.
\begin{proposition}[Augot-Zeh~\cite{AugotZeh_OnTheRREquations_2008}]
Given $s \geq 1$, let $Q(x,y)=\sum_{t=0}^{\ell} Q^{(t)}(x)y^t $ be the Guruswami-Sudan interpolation polynomial that satisfies~\refeq{eq_GS_cond1} and \refeq{eq_interpol2} and let $R(x)$ be the Lagrange interpolation polynomial, such that $R(\alpha_i)=r_i \; \forall i \in \SET{n}$ holds. Furthermore, let $G(x)=\prod_{j=1}^n(x-\alpha_j)$. Then, $Q(x,y)$ satisfies~\refeq{eq_interpol2}, if and only if there exist $s$ polynomials $B^{(b)}(x) \in \Fx \; \forall b \in \SETzero{s-1}$ with:
\begin{equation} \label{eq_GSKE}
Q^{[b]}(x,R(x)) = B^{(b)}(x) \cdot G(x)^{s-b},
\end{equation}
where $\deg B^{(b)}(x) < \ell(n-k) - s\tau+b$.
\end{proposition}
We remark that $Q^{[b]}(x,y) := Q^{[0,b]}(x,y) $ denotes the $b$--th Hasse derivative of the bivariate polynomial $Q(x,y)$ with respect to the variable $y$.

\section{Welch--Berlekamp approach as List--1 Decoder} \label{sec_listone}
We recall a simplified version (as in~\cite{Gemmell-Sudan_IFL1992} or~\cite[Ch. 5]{JustesenHoholdt_ACourseinError-CorrectingCodes_2004}) of the WB approach~\cite[Ch. 7.2]{Moon:ECCMMAM2005} \cite{BerlekampWelch_Patent} for decoding RS codes up to half the minimum distance ($\tau_0 = \lfloor (n-k)/2 \rfloor$). It is seen as special case of the list--decoding problem of GS.\\
The interpolation polynomial $Q(x,y)$ of the GS algorithm for $\ell=s=1$ has the following form:
\begin{equation*}
 Q(x,y) = \QPOL{0} + \QPOL{1}y,
\end{equation*}
where $\deg \QPOL{0} < n-\tau $ and $\deg \QPOL{1} < n-\tau-k+1$. Condition~\refeq{eq_interpol2} simplifies to 
$Q(\alpha_i, r_i) = 0 \; \forall i \in \SET{n}$ and gives $n$ linear equations.
The codeword $\mathbf{c}$ coincides with the received word $\mathbf{r}$ in
at least $n-\tau$ positions. Therefore, we have:
\begin{equation*}
 Q(x,f(x)) = \QPOL{0} + f(x) \cdot \QPOL{1} =  0.
\end{equation*}
So $f(x) = - \QPOL{0} / \QPOL{1}$ and we can rewrite the original interpolation polynomial:
\begin{equation*}
 Q(x,y) = \QPOL{1} \cdot \left( y + \frac{\QPOL{0}}{\QPOL{1}} \right) = \QPOL{1} \cdot (y-f(x)).
\end{equation*}
Clearly, \QPOL{1} is the error--locator polynomial (ELP), because it vanishes for $\tau_0$ $\alpha_i$'s. Let the classical ELP $\Lambda(x) = \prod_{j \in \mathcal{J}} (x-\alpha_j)$, where
$\mathcal{J}$ is the set of error locations. Then, we can write:
\begin{equation} \label{eq_WE_ELP}
 Q(x,y) = \Lambda(x) \cdot (y-f(x)).
\end{equation}
In the WB decoding procedure the polynomial $Q(x,y)$ of~\refeq{eq_WE_ELP} is determined by solving $n$ linear homogeneous equations.
The standard syndrome--based decoding procedure, that consists of $\tau_0$ equations for the ELP, can be derived by reducing the WB equation.

\section{Virtual Extension to an IRS code} \label{sec_pdec}
\subsection{Basic Principle}
We shortly describe the Schmidt--Sidorenko--Bossert scheme~\cite{Schmidt_DecodingReedSolomonCodesBeyondHalf_2006} where an RS code
is virtually extended to an IRS code. This IRS code is denoted by \VIRS{s}, where $n$ and $k$ are the original parameters
of the \RS{n}{k} code. The parameter $s$ denotes the order of interleaving. 
Let $p(x) = \sum_{j=0}^{n-1} p_j x^j$ be a univariate polynomial in $\Fxsub{n}$. Then, 
\begin{equation*}
 p^{<i>}(x) = \sum_{j=0}^{n-1} p_j^i x^j,
\end{equation*}
is the polynomial in $\Fxsub{n}$ where each coefficient is raised to the power $i$. Analogously, $\mathbf{c}^{<i>}$ denotes the vector $(c_1^i,\dots,c_n^i)^T$.
The virtual IRS code can be defined as follows.
\begin{definition}[Virtual Extension to an IRS code~\cite{Schmidt_DecodingReedSolomonCodesBeyondHalf_2006}]
Let \RS{n}{k} be an RS code with the evaluation polynomials $f(x)$ as defined in~\refeq{eq_defRS}. The virtually
extended Interleaved Reed--Solomon code $\VIRS{s}$ of order $s$ is given by
\begin{align*}
\VIRS{s} & =   
\begin{pmatrix}
\mathbf{c}^{<1>} \\
\mathbf{c}^{<2>} \\
\vdots \\
\mathbf{c}^{<s>}
\end{pmatrix} \\ 
& =  \begin{pmatrix}
f(\mathcal{L}) & : f(x) \in \Fxsub{k} \\
f^2(\mathcal{L}) & : (f(x))^2 \in \Fxsub{2(k-1)+1} \\
\vdots & \\
f^s(\mathcal{L}) & : (f(x))^s \in \Fxsub{s(k-1)+1} 
\end{pmatrix}.
\end{align*}
\end{definition}

Clearly, the parameter $s$ must satisfy $s(k-1) +1 \leq n$. The scheme is restricted to low--rate RS codes and allows to decode beyond half the minimum distance.
The virtual extension is illustrated in Figure~\ref{f_hetirs},
\begin{figure}[ht]
\centering
\includegraphics[angle=0, width=\columnwidth]{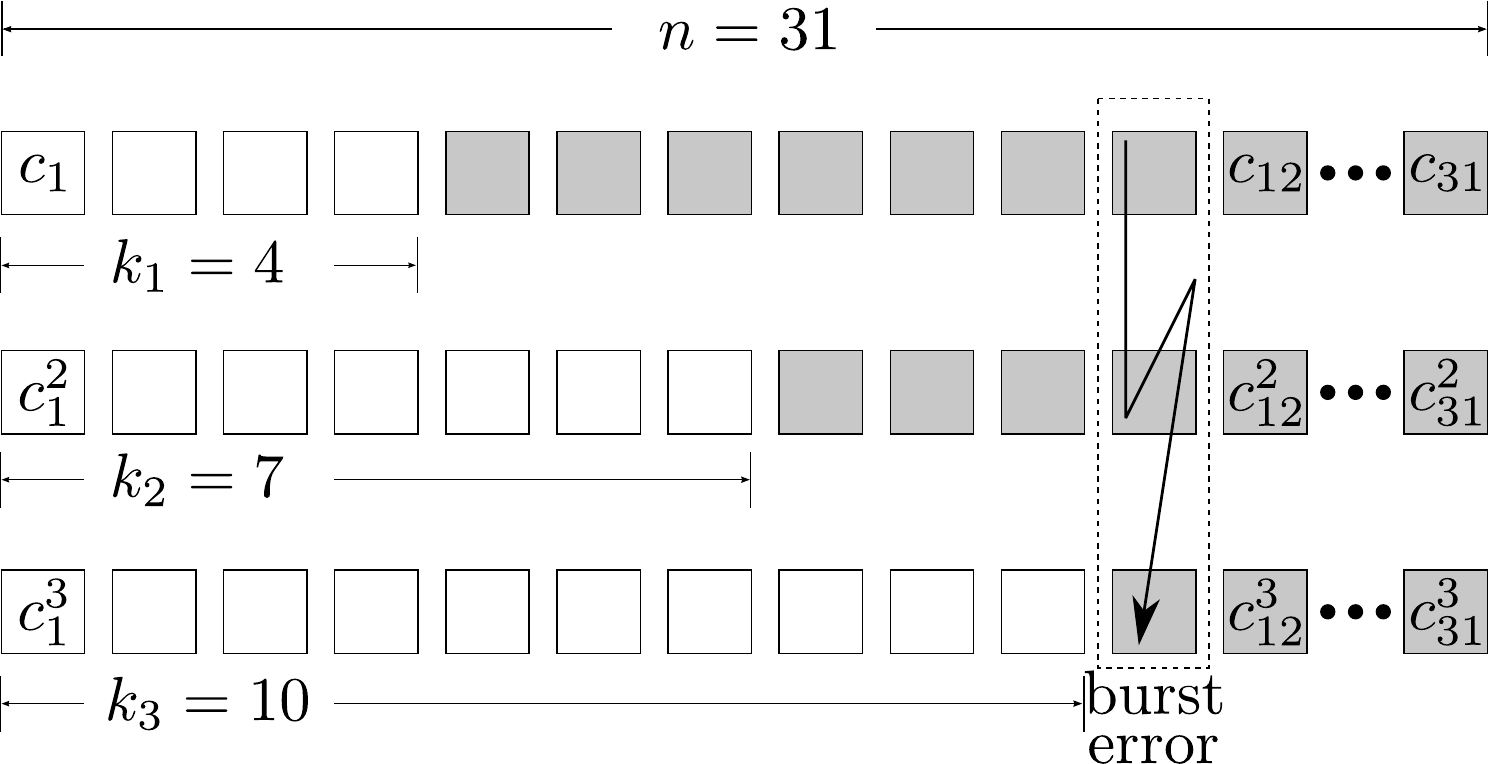}
\caption{Illustration of an $\mathcal{RS}(31,4)$ code that has been virtually extended with interleaving factor $s=3$. The errors in the $\mathcal{RS}(31,4)$ code are extended to burst errors in the $\mathcal{VIRS}(31,4,3)$ code.}
\label{f_hetirs}
\end{figure}
where the information length of the $i$--th codeword is $k^{(i)} = i(k-1)+1$. 
The decoding procedure for the virtual extension of an RS code is as follows; the elements of received word $\vect{r}=\vect{c}+\vect{e}$ are raised to the power $i=2,\dots,s$ ($\vect{r}^{<2>}, \vect{r}^{<3>}, \dots, \vect{r}^{<s>}$) and a heterogeneous IRS code is obtained. Clearly, through the virtual extension, the error is also ``extended'' and every single received word $\vect{r}^{<i>}$ is erroneous at the same positions. Due to the additional equations, the decoding radius is increased to:
\begin{equation} \label{eq_pdectmaxgivenl}
 \tau = \left\lfloor \frac{sn-\binom{s+1}{2}(k-1)-s}{s+1} \right\rfloor.
\end{equation}
The radius $\tau$ is greater than $\tau_0 =\lfloor (n-k)/2 \rfloor$ for RS codes with code rate $R < 1/3$.
(For further details (e.g. increased failure probability) of this scheme, see~\cite{Schmidt_DecodingReedSolomonCodesBeyondHalf_2006}).
We remark that the rate--restriction and the increased decoding radius coincide with the original Sudan algorithm (where the multiplicity $s$ equals one for all points $(\alpha_i,r_i)$). Nevertheless, we will show that this scheme is equivalent to a GS interpolation problem with a modified multiplicity assignment and stricter degree constraints.
To start the logical chain, we will describe in the following the corresponding system of equations of the $s$ WB equations for a $\VIRS{s}$ code.

\subsection{Matrix form of the Set of Equations}
Bleichenbacher \textit{et al.}~\cite{Bleichenbacher_DecodingofIRS_2003, Bleichenbacher_DecodingofIRS_2007} described the WB formulation for IRS codes.
We recall this approach for the virtually extended Reed--Solomon code $\VIRS{s}$.\\
Clearly, we have $s$ WB--equations (see \refeq{eq_WE_ELP}) of the form:
\begin{align} \label{eq_WB_VIRS}
  Q^{<b>}(x,y) & = \Lambda(x) \cdot (y^b-f^{b}(x)) \nonumber \\
 & =: \QPOL{s} y^b - \QPOL{b}, 
\end{align}
for all $b \in \SETzero{s-1}$.\\
For every single WB polynomial $Q^{<b>}(\alpha_i,r_i) = 0$ holds ($i \in \SET{n}$). Note, that through the virtual extension, each received word $\vect{r}^{<i>}$ has its errors at the same position and therefore we search one common ELP $\Lambda(x)$.
We represent the $sn$ constraints of system~\refeq{eq_WB_VIRS} in matrix form. Therefore, let the
$n \times (\tau+i(k-1)+1)$ matrix $\mat{M}_{i}$ be:
\begin{equation} \label{eq_elMatrix}
 \mat{M}_{i} = \begin{pmatrix}
 		1 &  	 \alpha_1 & \alpha_1^2  & \cdots & \alpha_1^{N_i-1} \\
 		1 &  	 \alpha_2 & \alpha_2^2  & \cdots & \alpha_2^{N_i-1} \\
 		1 &  	 \alpha_3 & \alpha_3^2  & \cdots & \alpha_3^{N_i-1} \\
 		\vdots & \vdots & \vdots & \ddots & \vdots \\
 		1 &  	 \alpha_n & \alpha_n^2  & \cdots & \alpha_n^{N_i-1}
 \end{pmatrix},
\end{equation}
where $N_i := \tau+i \cdot (k-1)+1$ and let $N$ be defined as:
\begin{equation} \label{eq_defineNsum}
 N = \sum_{i=0}^s N_i = (s+1)(\tau+1) + \binom{s+1}{2} (k-1).
\end{equation}
Furthermore, let the $n \times n$ matrix \mat{R} have the following form:
\begin{equation} \label{eq_elMatrix_rec}
 \mat{R} = \begin{pmatrix}
 		r_1 &  	 \cdots & 0  & 0  \\
 		0 &  	 r_2    & \cdots & 0 \\
 		\vdots & \vdots & \ddots & \vdots \\
		0 &  0	  & \cdots  & r_n \\
 \end{pmatrix}.
\end{equation}
Now, we can write the $s$ polynomial equations from~\refeq{eq_WB_VIRS} in matrix notation.
Let $\vect{Q} = (\vect{Q}^{(0)}, \vect{Q}^{(1)}, \dots, \vect{Q}^{(s)} )^T$, where
$\vect{Q}^{(i)} = (Q^{(i)}_0,Q^{(i)}_1,\dots, Q^{(i)}_{\tau +(s-i)(k-1)})^T$. The homogeneous set
of equations is of the form $\mat{A} \cdot \vect{Q} = \mathbf{0} $, where
the $sn \times N$ matrix $\mat{A}$ is:
\begin{equation} \label{eq_Matrix_pdec}
\mat{A} = \begin{pmatrix}
		\mat{0} & \cdots & \mat{0} &  - \mat{M}_{1} & \mat{R} \cdot \mat{M}_{0} \\
		\mat{0} & \cdots & - \mat{M}_{2} & \mat{0} & \mat{R}^{2} \cdot \mat{M}_{0} \\
 		\vdots & \iddots & \vdots & \vdots & \vdots \\
       -\mat{M}_{s} &  \mat{0} & \cdots & \mat{0} &  \mat{R}^{s} \cdot \mat{M}_{0} \\
\end{pmatrix}.
\end{equation}
The vector $\vect{Q}^{(s)}$ gives the coefficients of the ELP $\Lambda(x)$.

\section{Reformulation as a modified Guruswami--Sudan problem} \label{sec_GS-light}
\subsection{Specific Multiplicity Assignment}
In this section, we formulate the decoding of an \RS{n}{k} code virtually extended to a $\VIRS{s}$ code as a modified GS interpolation problem.
The constraints of the bivariate interpolation polynomial with multiplicities are a modified version of the general GS algorithm introduced
in Section~\ref{sec_GS}. We show the corresponding homogeneous set of equations and prove the equivalence to the one of Bleichenbacher \textit{et al.} (see~\refeq{eq_Matrix_pdec}).\\
Let $\overline{Q}(x,y)$ be a bivariate polynomial of $\Fxy \backslash \{0\} $, where
\begin{equation} \label{eq_GS-light_cond1}
DC_2 := \begin{bmatrix}
\deg_{0,1} \overline{Q}(x,y) \leq s,\\
\deg \overline{Q}^{(t)}(x) \leq \tau + (s-t) \cdot (k-1)  \\
\end{bmatrix}.
\end{equation} 
The modified interpolation constraints for $\overline{Q}(x,y)$ are:
\begin{equation} \label{eq_GS-light}
IC_2 := \begin{bmatrix} \overline{Q}^{[b]}(\alpha_i, r_i) = 0 \quad \forall i \in \SET{n} \; \text{and} \; \forall b \in \SETzero{s-1} \end{bmatrix},
\end{equation}
where the parameter $s$ is such that $s(k-1)+1 \leq n$ holds and $\overline{Q}^{[b]}(x,y)$ denotes the $b$--th Hasse derivative with respect to the variable $y$ of the polynomial $\overline{Q}(x,y)$.
\begin{theorem}
It exists at least one nonzero polynomial $\overline{Q}(x,y)$ which satisfies conditions~\refeq{eq_GS-light}.
\end{theorem}
\begin{proof}
Condition~\refeq{eq_GS-light} gives $sn$ homogeneous linear equations to the coefficients. The number of possible coefficients
is $N$ (as defined in~\refeq{eq_defineNsum}), therefore we get a nonzero solution for the decoding radius $\tau$ as in Equation~\refeq{eq_pdectmaxgivenl} .
\end{proof}
The bivariate polynomial $\overline{Q}(x,y)$ that fulfills condition~\refeq{eq_GS-light} has multiplicity $s$ for all $n-\tau$ error--free positions and
multiplicity one for all $\tau$ error positions. Let us state this property in the following theorem.
\begin{theorem} \label{th_GSlight1}
The bivariate polynomial $\overline{Q}(x,y)$ under the constraints $DC_2$ and $IC_2$ can be written as:
\begin{equation} \label{eq_explicit_GS_light}
 \overline{Q}(x,y) = \overline{Q}^{(s)}(x) \cdot (y-f(x))^s, 
\end{equation}
where $\overline{Q}^{(s)}(x)$ is the ELP and $f(x)$ is the information polynomial of the RS code (see definition~\refeq{eq_defRS}).
\end{theorem}
\begin{proof}
Let us consider the ``last'' $(s-1)$--th Hasse derivative of $\overline{Q}(x,y)$ with respect to the variable $y$:
\begin{align*}
 \overline{Q}^{[s-1]}(x,y) & = \binom{s-1}{s-1} \cdot  \overline{Q}^{(s-1)}(x) + \binom{s}{s-1} \cdot \overline{Q}^{(s)}(x)y \nonumber \\
 & = \overline{Q}^{(s-1)}(x) + s \cdot \overline{Q}^{(s)}(x)y \nonumber \\
 & = s \cdot \overline{Q}^{(s)}(x) \cdot \left(y + \frac{\overline{Q}^{(s-1)}(x)}{s \cdot \overline{Q}^{(s)}(x)} \right),
\end{align*}
which is by~\refeq{eq_GS-light} zero for the set $\lbrace(\alpha_i,r_i)\rbrace_{i=1}^n$. Clearly, $\overline{Q}^{[s-1]}(x,y)$ is a WB polynomial for the \RS{n}{k} code
with information polynomial $f(x) = - \overline{Q}^{(s-1)}(x)/ s \cdot \overline{Q}^{(s)}(x)$ (see Section~\ref{sec_listone}).\\
The $(s-2)$--th Hasse derivative of the interpolation polynomial $\overline{Q}^{[s-2]}(x,y)$ can now be rewritten as:
\begin{align} \label{eq_secondHasse}
  \overline{Q}^{[s-2]}(x,y) =  & \overline{Q}^{(s-2)}(x) + \binom{s-1}{s-2} \cdot \overline{Q}^{(s-1)}(x)y + \nonumber \\
  & \binom{s}{s-2} \cdot \overline{Q}^{(s)}(x)y^2 \nonumber \\
  = &  \overline{Q}^{(s-2)}(x) + (s-1) \cdot \overline{Q}^{(s-1)}(x)y + \nonumber \\
  & \frac{1}{2}s(s-1) \cdot \overline{Q}^{(s)}(x)y^2 \nonumber \\
  = &  \frac{1}{2}s(s-1) \cdot \overline{Q}^{(s)}(x) \cdot  (y^2 - f(x)y ) + \nonumber \\
  & \overline{Q}^{(s-2)}(x),
\end{align}
where
\begin{equation*}
 \overline{Q}^{[s-2]}(x,f(x)) = 0
\end{equation*}
from the interpolation constraints holds.
We can now express $\overline{Q}^{(s-2)}(x)$ as:
\begin{align}
 \overline{Q}^{(s-2)}(x) & = - \frac{1}{2}s(s-1) \cdot \overline{Q}^{(s)}(x) \cdot (f(x)^2 -2f(x)^2)  \nonumber \\
& = \frac{1}{2}s(s-1) \cdot f(x)^2 \cdot \overline{Q}^{(s)}(x).
\end{align}
Substituting this into~\refeq{eq_secondHasse}, we obtain for the $(s-2)$--th Hasse derivative of $\overline{Q}(x,y)$:
\begin{equation*}
 \overline{Q}^{[s-2]}(x,y) = \frac{1}{2}s(s-1) \cdot \overline{Q}^{(s)}(x) \cdot (y-f(x))^2,
\end{equation*}
which has multiplicity two at the $n-\tau$ error--free positions and multiplicity one at the $\tau$ erroneous positions. By induction we can state that $(y-f(x))^s \vert \overline{Q}(x,y)$ and $Q^{(s)}(x) \vert \overline{Q}(x,y)$. From $DC_2$ we know, that no other polynomial factor occurs in $\overline{Q}(x,y)$. 
\end{proof}
\subsection{Informal Description}
The degree condition $DC_2$ and the interpolation constraint $IC_2$ for the polynomial $\overline{Q}(x,y)$ are a subset of the general GS list--decoding constraints $DC_1$ and $IC_1$. The $y$--degree of $\overline{Q}(x,y)$ corresponds to the number of codewords of the $\VIRS{s}$ code.
Similar to the univariate formulation of the original GS interpolation problem (see~\refeq{eq_GSKE}) it is sufficient to consider only the Hasse derivatives with respect to variable $y$.\\ 
In the original GS algorithm the $b$--th Hasse derivative of the interpolation polynomial $Q(x,y)$ is divisible by $G(x)^{(s-b)}$, where $G(x)= \prod_{i=1}^n (x-\alpha_i)$ and $n$ denotes the code length. In our case
the $b$--th Hasse derivative of the modified interpolation polynomial $\overline{Q}(x,y)$ is divisible by $\overline{G}(x)^{(s-b)}$. Here, $\overline{G}(x) = \prod_{i \in \SET{n}\setminus \mathcal{J}} (x-\alpha_i)$ and $\SET{n}\setminus \mathcal{J}$ is the set of error--free positions.\\
Furthermore, the ELP $\overline{Q}^{(s)}(x)$, where $\deg \overline{Q}^{(s)}(x)$ can be greater than $\lfloor (n-k)/2 \rfloor$, is a factor of $\overline{Q}(x,y)$. The zeros of $\overline{Q}^{(s)}(x)$ have multiplicity one.\\
The scheme of Section~\ref{sec_pdec} virtually extends the received vector $\vect{r} = (r_1, r_2, \dots, r_n)$ of an \RS{n}{k} code to $s$ received words $\vect{r}^{<i>} = (r_1^i, r_2^i, \dots, r_n^i) \; \forall i \in \SET{s}$ of $s$ different \RS{n}{i(k-1)+1} codes with equal code length $n$.

\subsection{Set of Equations}
Now, we consider the homogeneous set of equations~\refeq{eq_GS-light}. We have $\mat{\overline{B}} \cdot \vect{\overline{Q}} = \mat{0}$, where $\vect{\overline{Q}}$ is the vector notation of the interpolation polynomial $\overline{Q}(x,y)$.
The $sn \times N$ matrix \mat{\overline{B}} can be written as: 
\begin{small}
\begin{equation*} \label{eq_mat_GS-light}
\left( \begin{array}{ccccc}
		\mat{0} & \cdots & \mat{0} &  \binom{s-1}{s-1} \mat{M}_{1} & \binom{s}{s-1} \mat{R} \mat{M}_{0} \\
		\mat{0} & \cdots & \binom{s-2}{s-2} \mat{M}_{2} & \binom{s-1}{s-2} \mat{R} \mat{M}_{1} & \binom{s}{s-2} \mat{R}^{2} \mat{M}_{0} \\
 		\vdots &  \iddots &  \vdots & \vdots & \vdots \\
    	\mat{M}_{s} & \mat{R} \mat{M}_{s-1}  & \cdots  & \mat{R}^{s-1} \mat{M}_{1} & \mat{R}^{s} \mat{M}_{0} \\
\end{array}\right) 
\end{equation*}
\end{small}
where the sub--matrices $\mat{M}_i$ and $\mat{R}$ are defined in~\refeq{eq_elMatrix} and~\refeq{eq_elMatrix_rec}. The binomial coefficients come from the Hasse derivatives of $\overline{Q}(x,y)$:
\begin{equation}
\overline{Q}^{[b]}(x,y) = \sum_{t=b}^s \binom{t}{b} \cdot \overline{Q}^{(t)}(x)y^{t-b}.
\end{equation}
Note, that the first $n$ rows of matrix $\mat{\overline{B}}$ correspond to the $(s-1)$th Hasse derivative of the polynomial $\overline{Q}(x,y)$.
The second $n$ rows represents the $n$ interpolation constraints of the $(s-2)$--th Hasse derivative and so on.
In the last $n$ rows of matrix $\mat{\overline{B}}$ the interpolation polynomial $\overline{Q}(x,y)$ occurs with all terms.

\subsection{Equivalence of Both Sets of Equations} \label{sec_equivalence}
In the following, we show the equivalence between the systems of equations determining the IRS scheme of Section~\ref{sec_pdec} and the one determining the modified GS interpolation polynomial $\overline{Q}(x,y)$. Due to space limitations we will sketch the basic steps of the proof.\\
First, let us consider the relation between vectors $\vect{Q}$ and $\vect{\overline{Q}}$:
\begin{align*}
 \overline{Q}(x,y) = & \overline{Q}^{(s)}(x) (y-f(x))^s \nonumber  \\
  = & \overline{Q}^{(s)}(x) \cdot \left( \sum_{i=0}^s \binom{s}{i} (-1)^i y^{s-i}f(x)^i \right).
\end{align*}
In vector notation, we have:
\begin{multline*}
 (\vect{Q}^{(0)}, \dots, \binom{s}{s-2} \vect{Q}^{(s-2)}, - \binom{s}{s-1}\vect{Q}^{(s-1)}, \vect{Q}^{(s)} )^T = \\ 
 (\vect{\overline{Q}}^{(0)}, \dots, \vect{\overline{Q}}^{(s-2)}, \vect{\overline{Q}}^{(s-1)}, \vect{\overline{Q}}^{(s)} )^T.
\end{multline*}
Let the matrix $\mat{B}$ be such that:
\begin{equation*}
 \mat{\overline{B}} \cdot \vect{\overline{Q}} = \mat{B} \cdot \vect{Q}.
\end{equation*}
Matrix $\mat{B}$ is then (first column not printed):
\begin{small}
\begin{multline*}
\mat{B} =  \\
\left( \begin{array}{cccc}
		\cdots & \mat{0} & -\binom{s}{s-1} \mat{M}_{1} & \binom{s}{s-1} \mat{R} \mat{M}_{0} \\
		\cdots & \binom{s}{s-2} \mat{M}_{2} & -\binom{s}{s-1} \binom{s-1}{s-2} \mat{R} \mat{M}_{1} & \binom{s}{s-2} \mat{R}^{2} \mat{M}_{0} \\
 		\iddots &  \vdots & \vdots & \vdots \\
    	\cdots  & \cdots  & -\binom{s}{s-1} \mat{R}^{s-1} \mat{M}_{1} & \mat{R}^{s} \mat{M}_{0} \\
\end{array}\right).
\end{multline*}
\end{small}
After simplification, we obtain:
\begin{small}
\begin{multline*} 
\mat{B} =  \\
\left( \begin{array}{cccc}
		\cdots & \mat{0} &  - \mat{M}_{1} & \mat{R} \mat{M}_{0} \\
		\cdots & \frac{1}{2}s(s-1) \mat{M}_{2} & -s(s-1) \mat{R} \mat{M}_{1} & \frac{1}{2}s(s-1) \mat{R}^2 \mat{M}_{0} \\
 		\iddots &  \vdots & \vdots & \vdots \\
    	\cdots  & \cdots  & -s \mat{R}^{s-1} \mat{M}_{1} & \mat{R}^{s} \mat{M}_{0} \\
\end{array}\right).
\end{multline*}
\end{small}
The second band of $n$ rows of matrix $\mat{B}$ can be multiplied with $-\mat{R}s(s-1)$--times the first band of $\mat{B}$ and then divided by $-\frac{1}{2}s(s-1)$. We obtain 
the second band of $n$ rows of matrix $\mat{A}$. Repeating this operation, matrix $\mat{B}$ can be transformed into matrix $\mat{A}$~\refeq{eq_Matrix_pdec}.

\section{Conclusion} \label{sec_future}
We investigated a virtual extension of an RS code to an IRS code from an interpolation--based list--decoding approach point of view.\\
The Bleichenbacher scheme was used to form the system of equations for the IRS scheme (based on a virtual extension). Then, the original constraints of the GS list--decoding algorithm were modified and the equivalence of the resulting system of equations with the Bleichenbacher scheme for the IRS code has been shown.


\begin{thebibliography}{10}
\providecommand{\url}[1]{#1}
\csname url@samestyle\endcsname
\providecommand{\newblock}{\relax}
\providecommand{\bibinfo}[2]{#2}
\providecommand{\BIBentrySTDinterwordspacing}{\spaceskip=0pt\relax}
\providecommand{\BIBentryALTinterwordstretchfactor}{4}
\providecommand{\BIBentryALTinterwordspacing}{\spaceskip=\fontdimen2\font plus
\BIBentryALTinterwordstretchfactor\fontdimen3\font minus
  \fontdimen4\font\relax}
\providecommand{\BIBforeignlanguage}[2]{{%
\expandafter\ifx\csname l@#1\endcsname\relax
\typeout{** WARNING: IEEEtranS.bst: No hyphenation pattern has been}%
\typeout{** loaded for the language `#1'. Using the pattern for}%
\typeout{** the default language instead.}%
\else
\language=\csname l@#1\endcsname
\fi
#2}}
\providecommand{\BIBdecl}{\relax}
\BIBdecl

\bibitem{AugotZeh_OnTheRREquations_2008}
\BIBentryALTinterwordspacing
D.~Augot and A.~Zeh, ``On the {R}oth and {R}uckenstein {E}quations for the
  {G}uruswami-{S}udan {A}lgorithm,'' in \emph{Information Theory, 2008. ISIT
  2008. IEEE International Symposium on}, 2008, pp. 2620--2624. [Online].
  Available: \url{http://dx.doi.org/10.1109/ISIT.2008.4595466}
\BIBentrySTDinterwordspacing

\bibitem{BerlekampWelch_Patent}
E.~R. Berlekamp and L.~Welch, ``Error correction of algebraic block codes,'' US
  Patent Number 4,633,470.

\bibitem{Bleichenbacher_DecodingofIRS_2003}
\BIBentryALTinterwordspacing
D.~Bleichenbacher, A.~Kiayias, and M.~Yung, ``{Decoding of Interleaved
  Reed--Solomon Codes over Noisy Data},'' in \emph{Automata, Languages and
  Programming}, ser. Lecture Notes in Computer Science, 2003, ch.~9, p. 188.
  [Online]. Available: \url{http://dx.doi.org/10.1007/3-540-45061-0\_9}
\BIBentrySTDinterwordspacing

\bibitem{Bleichenbacher_DecodingofIRS_2007}
\BIBentryALTinterwordspacing
------, ``{Decoding Interleaved Reed--Solomon codes over noisy channels},''
  \emph{Theor. Comput. Sci.}, vol. 379, no.~3, pp. 348--360, 2007. [Online].
  Available: \url{http://dx.doi.org/10.1016/j.tcs.2007.02.043}
\BIBentrySTDinterwordspacing

\bibitem{Gemmell-Sudan_IFL1992}
\BIBentryALTinterwordspacing
P.~Gemmell and M.~Sudan, ``Highly resilient correctors for polynomials,''
  \emph{Informartion Processing Letters}, vol.~43, no.~4, pp. 169--174, 1992.
  [Online]. Available: \url{http://dx.doi.org/10.1016/0020-0190(92)90195-2}
\BIBentrySTDinterwordspacing

\bibitem{GuruswamiSudan_ImproveddecodingofReed-Solomonandalgebraic-geometrycod%
es_1999}
\BIBentryALTinterwordspacing
V.~Guruswami and M.~Sudan, ``Improved decoding of {R}eed-{S}olomon and
  algebraic-geometry codes,'' \emph{IEEE Transactions on Information Theory},
  vol.~45, no.~6, pp. 1757--1767, 1999. [Online]. Available:
  \url{http://ieeexplore.ieee.org/xpls/abs\_all.jsp?arnumber=782097}
\BIBentrySTDinterwordspacing

\bibitem{Hasse_1936}
H.~Hasse, ``Theorie der h\"{o}heren {D}ifferentiale in einem algebraischen
  {F}unktionenk\"{o}rper mit vollkommenem {K}onstantenk\"{o}rper bei beliebiger
  {C}harakteristik.'' \emph{J. Reine Angew. Math.}, vol. 175, pp. 50--54, 1936.

\bibitem{JustesenHoholdt_ACourseinError-CorrectingCodes_2004}
J.~Justesen and T.~Hoholdt, \emph{A Course in Error-Correcting Codes (EMS
  Textbooks in Mathematics)}.\hskip 1em plus 0.5em minus 0.4em\relax European
  Mathematical Society, February 2004.

\bibitem{Krachkovsky_IRS_1998}
\BIBentryALTinterwordspacing
Krachkovsky and Y.~X. Lee, ``{Decoding of parallel Reed--Solomon codes with
  applications to product and concatenated codes},'' August 1998, pp. 55+.
  [Online]. Available: \url{http://dx.doi.org/10.1109/ISIT.1998.708636}
\BIBentrySTDinterwordspacing

\bibitem{Moon:ECCMMAM2005}
T.~K. Moon, \emph{Error Correction Coding: Mathematical Methods and
  Algorithms}.\hskip 1em plus 0.5em minus 0.4em\relax Wiley-Interscience, June
  2005.

\bibitem{Roth_Ruckenstein_2000}
\BIBentryALTinterwordspacing
R.~M. Roth and G.~Ruckenstein, ``Efficient decoding of {R}eed-{S}olomon codes
  beyond half the minimum distance,'' \emph{Information Theory, IEEE
  Transactions on}, vol.~46, no.~1, pp. 246--257, 2000. [Online]. Available:
  \url{http://ieeexplore.ieee.org/xpls/abs\_all.jsp?arnumber=817522}
\BIBentrySTDinterwordspacing

\bibitem{Ruckenstein_PHD2001}
\BIBentryALTinterwordspacing
G.~Ruckenstein, ``Error decoding strategies for algebraic codes,'' Ph.D.
  dissertation, Technion, 2001. [Online]. Available:
  \url{http://www.cs.technion.ac.il/users/wwwb/cgi-bin/tr-info.cgi/2001/PHD/PH%
D-2001-01}
\BIBentrySTDinterwordspacing

\bibitem{Schmidt_DecodingReedSolomonCodesBeyondHalf_2006}
\BIBentryALTinterwordspacing
G.~Schmidt, V.~Sidorenko, and M.~Bossert, ``Decoding {R}eed-{S}olomon {C}odes
  {B}eyond {H}alf the {M}inimum {D}istance using {S}hift-{R}egister
  {S}ynthesis,'' in \emph{Information Theory, 2006 IEEE International Symposium
  on}, 2006, pp. 459--463. [Online]. Available:
  \url{http://ieeexplore.ieee.org/xpls/abs\_all.jsp?arnumber=4036003}
\BIBentrySTDinterwordspacing

\bibitem{sudan97decoding}
\BIBentryALTinterwordspacing
M.~Sudan, ``Decoding of reed solomon codes beyond the error-correction bound,''
  \emph{Journal of Complexity}, vol.~13, no.~1, pp. 180--193, March 1997.
  [Online]. Available: \url{http://dx.doi.org/10.1006/jcom.1997.0439}
\BIBentrySTDinterwordspacing

\bibitem{ZehA_AnBerlekampMassey_2010}
A.~Zeh, C.~Gentner, and D.~Augot, ``{A Berlekamp-Massey Approach for the
  Guruswami-Sudan Decoding Algorithm for Reed-Solomon Codes},''
  \emph{preprint}, 2010.


\end{thebibliography}


\appendix
Let us consider an \RS{16}{4} code over $\F=GF(17)$ with parameter $s=2$ (number of interleaving and multiplicity for the modified GS algorithm). The corresponding increased decoding radius is $\tau=7$ (see~\refeq{eq_pdectmaxgivenl}).\\
The code locators are $\alpha_i = \alpha^{i-1} \ \forall i \in \SET{n}$, where $\alpha$ is $3$. For the information polynomial $f(x) = 1+x+x^2+x^3$ (see~\refeq{eq_defRS}) and an error $\vect{e}$ of weight $\tau = 7$ we get the following vectors: 
\begin{align*} \label{eq_Ex1}
	\vect{c} & = (4, 6, 4, 6, 0, 3, 12, 2, 0, 14, 7, 9, 0, 15, 15, 4)\\
	\vect{e} & = (1, 2, 3, 4, 5, 6, 7, 0, 0, 0, 0, 0, 0, 0, 0, 0)\\
	\vect{r}^{<1>} & = (5, 8, 7, 10, 5, 9, 2, 2, 0, 14, 7, 9, 0, 15, 15, 4) \\
	\vect{r}^{<2>} & = (8, 13, 15, 15, 8, 13, 4, 4, 0, 9, 15, 13, 0, 4, 4, 16).
\end{align*}
The conditions~\refeq{eq_GS-light} on the modified bivariate polynomial $\overline{Q}(x,y)$ give the following solution:
\begin{align*}
\vect{\overline{Q}} = & (5,14,8,6,14,9,5,9,12,12,4,2,3,16,5,9,11,\\
 & 15,13,4,7,2,16,4,16,5,4,2,4,3,12,5,16).
\end{align*}
And the corresponding modified bivariate interpolation polynomial;
\begin{align*}
 \overline{Q}(x,y) = & (x + 2)(x + 4)(x + 7)(x + 8)(x + 12)(x + 14) \cdot \\
 & (x + 16)(y + 16x^3 + 16x^2 + 16x + 16)^2, 
\end{align*}
is factorizable as stated in Theorem~\ref{th_GSlight1}.
\end{document}